\documentclass[sigconf]{acmart}

\copyrightyear{2020}
\acmYear{2020}
\acmConference[WWW '20]{Proceedings of The Web Conference 2020}{April 20--24, 2020}{Taipei, Taiwan}
\setcopyright{iw3c2w3}
\acmBooktitle{Proceedings of The Web Conference 2020 (WWW '20), April 20--24, 2020, Taipei, Taiwan}
\acmPrice{}
\acmDOI{10.1145/3366423.3380276}
\acmISBN{978-1-4503-7023-3/20/04}

\usepackage{balance}

\usepackage{amsmath}
\usepackage{amssymb}

\usepackage{graphicx}

\usepackage{url}
\urldef{\mailgr}\path|geoff.ramseyer@cs.stanford.edu|
\urldef{\mailag}\path|ashishg@stanford.edu|

\graphicspath{ {.} }

\begin{document}
\title {Liquidity in Credit Networks with Constrained Agents}

\author{Geoffrey Ramseyer}
\email{geoff.ramseyer@cs.stanford.edu}
\affiliation{%
  \institution{Stanford University}
}

\author{Ashish Goel}

\email{ashishg@stanford.edu}
\affiliation{%
  \institution{Stanford University}
}

\author {David Mazi{\`e}res}
\affiliation{%
  \institution{Stanford University}
}

\begin{abstract}
  In order to scale transaction rates for deployment across the
  global web, many cryptocurrencies have deployed so-called "Layer-2"
  networks of private payment channels.  An idealized payment network
  behaves like a {\em Credit Network}, a model for transactions across a
  network of bilateral trust relationships. Credit
  Networks capture many aspects of traditional currencies as well as
  new virtual currencies and payment mechanisms. In the traditional credit network model,
  if an agent defaults, every other node that trusted it is vulnerable
  to loss.  In a cryptocurrency context, trust is manufactured by capital deposits, and thus there arises a
  natural tradeoff between network liquidity (i.e. the fraction of transactions that succeed) and the cost of capital deposits.

  In this paper, we introduce constraints that bound the
  total amount of loss that the rest of the network can suffer if an
  agent (or a set of agents) were to default - equivalently, how the network changes
  if agents can support limited solvency guarantees.

  We show that these constraints preserve the analytical structure of a credit network.  
  Furthermore, we show that aggregate borrowing constraints greatly simplify the network structure and 
  in the payment network context achieve the optimal tradeoff between liquidity and amount of escrowed capital.


\end{abstract}

\ccsdesc[500]{Applied computing~Electronic funds transfer}
\ccsdesc[300]{Applied computing~Digital cash}
\ccsdesc[100]{Mathematics of computing~Network flows}
\keywords{Trust, Credit Networks, Electronic Fund Transfer}

\maketitle

\begin{acks}

The authors would like to thank the anonymous reviewers for their helpful suggestions on the presentation of this work.  This work was supported by the Stanford Center for Blockchain Research.

\end{acks}

\section{Introduction}

Practical implementations of markets require easily useable liquid currency.  But as transactions become larger and more frequent, moving and storing a physical asset, like gold or dollar bills, becomes very expensive.

Instead, people transfer money via promises to pay later.  Banks, for example, used to issue physical bank notes in exchange for gold deposits.  Individuals write checks to each other.  Some non-governmental organizations issue their own currencies.  Retailers issue prepaid gift cards.

Importantly, these debt notes are tradeable independent of the original issuer as a form of currency.  But to a first approximation, rational individuals will only accept a debt note if they trust the original issuer to satisfy the obligation.\footnote{Or perhaps if the individual expects that others expect that the issuer will satisfy the obligation.  This does not materially affect any of our analysis.}  For example, bank notes are usable in lieu of deposited currency, but only so long as the bank does not collapse.  In the United States, the Federal Deposit Insurance Corporation ensures that every bank's notes are redeemable in quantities up to \$250,000.  This enables individuals who use different banks to transact freely, but individuals might want to make sure that they do not have more than \$250,000 deposited in a single bank.

In order to send a payment, then, a payer needs to exchange her notes for notes acceptable to the payee.   Moreover, careful individuals might track the total value of notes owned from a single issuer, to mitigate exposure to the default of a single agent.  

However, our scenario need not solely consist of consumers transacting using notes issued by large institutions.  Individuals can also issue their own debt obligations, via, for example, checks or interpersonal promises, and furthermore can trade these with each other.  Generally speaking, most individuals trust their friends to repay small debts, but might worry about being repaid if one friend repeatedly tries to borrow lots of money.

\begin{figure}

\centering
\caption{A credit network.}
\includegraphics[width=8cm]{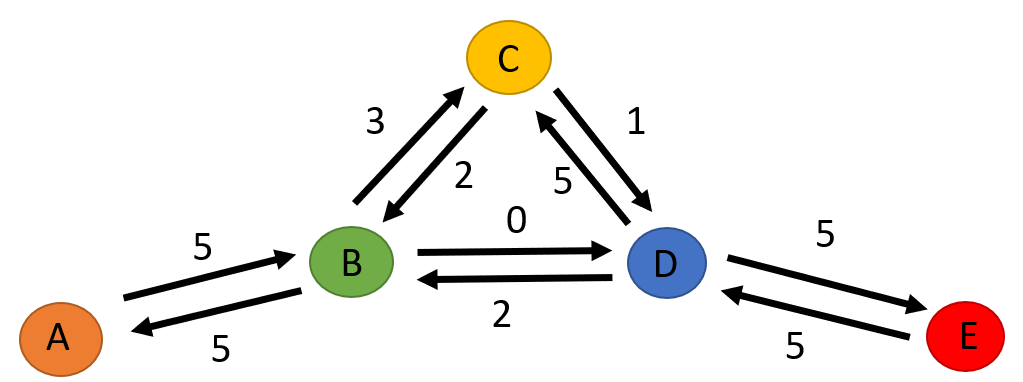}
\Description{A graph with 5 nodes.  Some nodes are connected by a pair of directed, weighted edges.}

\label{fig:cn}
\end{figure}

Formally, then, consider a system in which agents $u_1,...,u_n$ are trading debt notes, and where individual $u_i$ is willing to accept $w(u_i, u_j)$ notes from individual $u_j$.  The graph in figure \ref{fig:cn} provides a visual representation of a credit network with $5$ agents, where an arrow from agent $A$ to agent $B$ labeled with $5$ means that $w(A,B)=5$.  As a real-world example, in an economically healthy country, individuals are typically willing to accept functionally infinite numbers of notes from the central bank, and can transact by trading these notes.  The resulting network appears star-like, with the central bank at the center.  

Such a model is known as a {\em Credit Network}.  In a general credit network, autonomous agents can issue their own notes, and other agents can choose whether to accept these notes as payment, i.e. they can decide whether to trust any other agent, and for how much.  Money is sent along paths of trust, and reduces the amount of ``residual trust'' along the path; the transaction fails if no path exists. Trust is replenished by a transaction in the opposite direction.  Ghosh et al, \cite{ghosh2007mechanism}, De Figueiredo et al,\cite{defigueiredo2005trustdavis} and Karlan et al\cite{karlan2009trust} independently formulated this model, and Dandekar et al \cite{dandekar2011liquidity} formalized the model's mathematics.  Credit networks have also been used in practice for applications built on existing trust networks.  Examples include P2P systems that enable trading goods across a social network \cite{liu2010p2p}, Ostra \cite{mislove2008ostra}, a system to combat email spam, and the Yootle \cite{reeves2007yootopia}, a currency system for quantifying utility in group decision making.

More recently, credit networks are in use to improve cryptocurrency transaction rate and latency. In most blockchains, all network participants must agree on a global shared state, which limits transaction rate and can cause hours of latency.

Cryptocurrencies enable trustless, anonymous transactions.  But in reality, some pairs of agents might know each other and wish to transact repeatedly.  If these pairs trusted each other, they could transact without putting any information on a blockchain.  Instead, they could privately track the net balance of their transactions and settle this balance only as necessary.

One innovation of the Lightning network \cite{poon2016bitcoin} and analogous ``Layer 2'' networks on other cryptocurrencies is a way of using escrow to build bilateral relationships that are analogous to a traditional credit network's trust-based transaction channels, {\it without} actually requiring real trust for solvency.  Individuals need not put every transaction on the blockchain; rather, they need only to threaten to put their transactions on the blockchain.  Such threats are made credible if two parties put money into escrow on the blockchain, and the net balance between the parties does not exceed the amount of money in escrow.  The result is a large network of private channels of specific ``trust'' capacities, where transactions can route along paths in the graph.  Lightning, therefore, is exactly an implementation of a credit network.

\begin{figure}
\centering

\caption{A Lightning-style network, equivalent to the credit network in figure \ref{fig:cn}.}
\includegraphics[width=8cm]{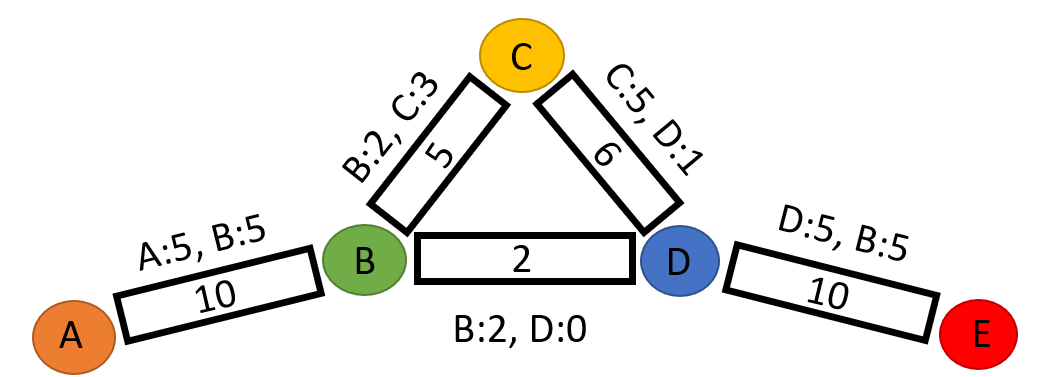}
\Description{A graph where undirected edges have associated escrow accounts and ownership records}

\label{fig:ln}
\end{figure}

For example, consider the Lightning-style network in figure \ref{fig:ln}.  Every undirected edge has, in cryptocurrency parlance, a fixed amount of capital in escrow.  The two parties to an edge possess certificates that record how much of the escrowed capital belongs to each party, and one party can ``pay'' the other by altering this balance.  No party can own more than 100\% of the escrowed capital.  Thus, if $A$ owns $5$ units of the money escrowed on edge $(A, B)$, $B$ can accept up to but no more than $5$ units of money from $A$.  We can model the Lightning network, therefore, as a credit network where, in this case, $w(B,A)=5$.  In fact, the set of transactions possible in Lightning-style network in figure \ref{fig:ln} are exactly those possible in the credit network in figure \ref{fig:cn}.

A given credit network cannot always resolve every possible transaction.  A node that is bankrupt, for example, cannot send additional money.  In cryptocurrency applications, when a transaction cannot be resolved by an overlaid credit network, it typically must instead be resolved on-chain.  On-chain transaction rates are typically quite limited.  The ratio between the number of transactions that a credit network can resolve to the number that it cannot, therefore, acts as a multiplier on the effective transaction rate of a blockchain.

However, placing capital into escrow to secure an edge is expensive.  Typically, more capital in escrow means a higher probability of transaction success (henceforth, liquidity, i.e. the fraction of transactions which succeed, given an exogenous transaction distribution), so agents must balance liquidity against escrow costs.


In this paper, we study how global guarantees on agent behavior (beyond the bilateral trading restrictions from the credit network) can alter the operation of a credit network.  Just as the guarantees on debt fulfillment provided by the FDIC streamline real-world transactions, repayment guarantees in network lending contexts or in cryptocurrency contexts (equivalently, restrictions on an agent's global borrowing) can achieve the optimal tradeoff between liquidity and escrowed capital.


Formally, we study the liquidity of a credit network in which every node is disallowed from borrowing more than some quantity in aggregate from its neighbors; we call these {\em node constraints}.  More generally, we study constraints on the total amount of ``net borrowing'' between any set of nodes and the rest of the network; we call these {\em predicate constraints}. In addition to being natural in their own right, they have specific advantages in many real-life situations that are well modeled by credit networks. We give three examples:
\begin{enumerate}
\item In Lightning, aggregate node constraint would allow pairwise relationships to be truly trust-based and not based on pairwise escrow; each node could be subject to an aggregate node constraint, and secure its relationships by putting just the aggregate amount in escrow.  Such a system can be implemented via multi-party smart contracts.  As we show here, when every node has such a constraint, the system as a whole achieves the optimal tradeoff between liquidity and escrow costs.
\item The popular app SplitWise~\cite{splitwise} allows a group of friends to track shared expenses. A process called ``simplify your debts'' cancels debts along cycles. This ``cycle-canceling'' is an essential aspect of credit networks, and SplitWise can be modeled as a credit network with infinite trust capacities. We believe that node constraints will greatly increase the robustness and usefulness of SplitWise, without substantially decreasing liquidity.
\item   The cryptocurrency Stellar~\cite{mazieres2015stellar} uses credit networks in two different ways.  First, it allows ``anchor'' nodes to
        issue tokens representing claims on fiat currency.
        Users can then issue ``trustlines'' declaring how much of each token they are willing to hold.  
        The resulting network of issued notes and lines of trust is very close to a credit network.  Note that Stellar allows 
        token issuers to lock the issuing account permanently.  This fixes the supply of a token, in effect implementing a node constraint.  And second, Stellar is in the process of building a Layer 2 protocol like Lightning.  
        Greater liquidity in this network would mean cheaper payments and forex trades.  
\end{enumerate}

In \cite{dandekar2011liquidity}, Dandekar et al analyze the liquidity of a network for a few classes of graphs of interest, and use computational simulations to conjecture liquidity when analysis is intractable.  In this work, we extend their results to new classes of networks that can model agent behavior under interesting classes of constraints.  Constraints break an analytical tool fundamental to \cite{dandekar2011liquidity,goel2015connectivity}; we show here how to analyze credit networks and their constrained variants with a new set of analytical tools.  

In section 4, we analyze the liquidity of several natural classes of constrained graphs, showing a tight connection between edge expansion and liquidity.  We then show that imposing node constraints not only preserves liquidity but also simplifies network structure and achieves the optimal tradeoff between liquidity and number of edges.

As an example, any network that extends to two agents $u$ and $v$ $D$ total units of credit has liquidity between that pair at most $1-1/D$.  Note that the graph that achieves this consists of $D$ parallel edges between $u$ and $v$, and thus the liquidity between $u$ and $w\neq v$ is $0$.  In a $d$-regular graph with edge expansion $\beta$ (where edges have capacity 1 and the transaction matrix is uniform), then the total credit available to each node is $d$ but the best known bounds \cite{goel2015connectivity} give liquidity only $1-2/\beta$ on average.  If nodes are constrained to borrow or lend at most $\beta/2$, then the total credit available to each node is $\beta$ and the pairwise liquidity lies between $1-1\beta$ and $1-2/\beta$, achieving the optimal liquidity tradeoff for {\it every pair simultaneously}.

Finally, we remark on some applications to Lightning, particularly how this tightened tradeoff can substantially reduce Lightning's escrow costs, and some open problems related to credit networks.  

\section{The Credit Network Model}

A configuration of a {\em Credit Network} is a directed graph $G=(V,E)$ along with a map $w((u,v))\geq 0$ denoting the amount of $v$'s currency that a node $u$ is willing to accept from $v$. In this article, all credit values will be integral.  For convenience, we say that if an edge $(u,v)\notin E$, then $w(u,v)=0$ and vice versa.

Suppose that agents $u$ and $v$ are transacting only with each other, and suppose $u$ tries to send one unit of its currency to $v$.  If $w(v,u)=0$, then $v$ is unwilling to accept the note from $u$, and the transaction fails.  But if $w(v,u)=k>0$, then $v$ is willing to accept the note.  Afterwards, $v$ will be only willing to accept an additional $k-1$ notes from $u$, and thus $w(v,u)$ decreases by $1$. Conversely, $v$ now owns one note from $u$ that $u$ must honor, and thus could send $w(u,v)+1$ total notes to $u$.  Hence, $w(u,v)$ increases by $1$, and the total trust $c(u,v)=w(u,v)+w(v,u)$ is constant.  As such, we can refer to a {\em Credit Network} as the undirected analogue of a configuration.  Note that a credit network has many configurations.


We call a transaction between neighbors as above a one-hop transaction.  More generally, multi-hop transaction of value $X$ is a payer $u$, a payee $v$, and a path $(p_0=u, p_1,...,p_t=v)$ from $u$ to $v$.  The transaction is valid if $w(p_i, p_{i+1})\geq X$ for $0\leq i<t$, and performing the transaction means performing a one-hop transaction of value $X$ along every edge $(p_i, p_{i+1})$.  This process is analogous to performing an augmenting path update in a max-flow computation.  For example, the configuration of figure \ref{fig:cn2} is the result of routing one unit from A to E along the route A-B-C-D-E, starting at the configuration in figure \ref{fig:cn}.\footnote{In this paper, we assume that all nodes issue notes in the same denomination.  It does not materially change results to convert different valuations into a common unit of value, as noted in \cite{dandekar2011liquidity}.}

\begin{figure}

\centering
\caption{The credit network of figure \ref{fig:cn}, after A has routed one unit of payment to E.}
\includegraphics[width=8cm]{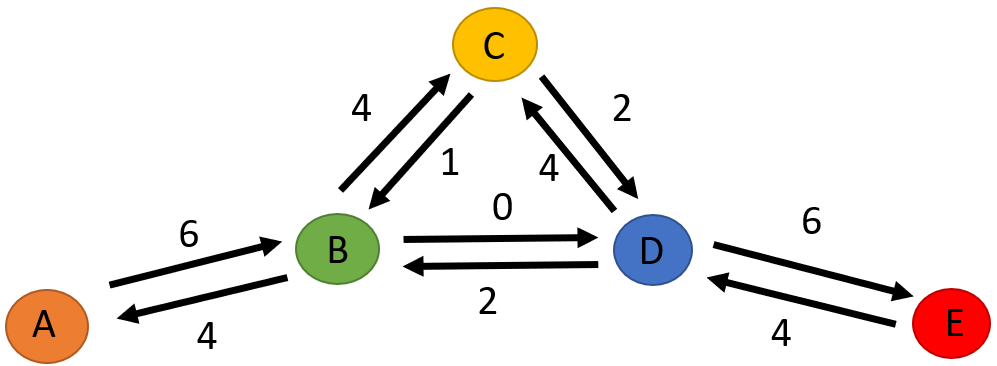}
\Description{The same graph as in \ref{fig:cn}, but with different weights on edges, reflecting that one node has routed payment to another}

\label{fig:cn2}
\end{figure}

Depending on context, we may refer to an edge in a credit network as "trust," or "possesses note", or "net borrowing."  These should be thought of as equivalent.  Every agent trusts the value of the notes they issue, so if some other agent $v$ possesses agent $u$'s issued note, then agent $u$ necessarily trusts that $v$ can send one unit of payment back to $u$.  Net borrowing or lending is relative to some ground state.  However, if we declare that a particular state in some implementation is the neutral state where nobody has transferred any debt notes, then an agent $u$'s "net lending" is the net amount of notes that $u$ has transferred to others - equivalently, the net borrowing is how much of the total trust capacity (in the ground state) that $u$ has used.

\subsection{Properties of Credit Networks}

For use as a payment method, agents care primarily about whether money can be sent in the current network configuration.  The particular details of a configuration in question matter far less.  This suggests the following definition:

\begin{definition}[Transaction Equivalence]

Two configurations of a credit network $C_1$ and $C_2$ are transaction-equivalent if for any list $\tau$ of transactions, all of transactions in $\tau$ can be successfully performed in sequence if the credit network starts at $C_1$ if and only if they can be performed starting at $C_2$.

\end{definition}

This definition will be useful later, but is not immediately useful for understanding the structure of the space of credit network configurations. 

Consider as a demonstrative example a cycle on $n$ vertices where each edge has capacity $1$, and the configuration where all edges have capacity $1$ in the direction towards a vertex $y$ and away from a vertex $x$.  Then clearly $y$ can route $1$ unit of money to $x$ by two distinct routes, by routing either ``clockwise'' or ``counterclockwise''.  After routing such a payment, all edges will be oriented either clockwise or counterclockwise.  Then, for any other vertices $w$ and $z$, no matter which route $y$ chose, $w$ can route exactly one unit of money to $z$.

In fact, the configurations where all edges are routed either clockwise or counterclockwise are transaction equivalent.  Moreover, if in one of these configurations, a vertex routes a payment to itself along the cycle, the network will reach the other configuration of the pair.  This motivates the following definition.

\begin{definition}[Cycle Equivalence (Definition 1, \cite{dandekar2011liquidity}]

Two configurations are cycle-equivalent if and only if one is reachable from the other by routing payments along cycles.

\end{definition}

In the above example, two configurations are cycle-equivalent if and only if they are also transaction-equivalent.  In fact, this correspondence holds for general graphs.
\begin{lemma}[(Lemma 2, \cite{dandekar2011liquidity})]

Two credit network configurations $C_1$ and $C_2$ are transaction-equivalent if and only if $C_1$ and $C_2$ are cycle-equivalent.

\end{lemma}

An analogue of this lemma also applies configuration changes resulting from processing a transaction.

\begin{lemma}[Route Independence (Theorem 3, \cite{dandekar2011liquidity})]
The cycle-equivalence class that results from routing a payment from a vertex $x$ to a vertex $y$ starting at some configuration $C$ is constant no matter the choice of route.
\end{lemma}

In a broad sense, our object of study is the performance of a credit network, in terms of its ability to resolve attempted transactions.  In most contexts where a credit network could be used, transactions arise from some exogenous process.  For example, many cryptocurrency transactions arise from online commerce or in response to real-world price fluctuations, not from the internal state of the Lightning network.

We will define the {\em liquidity} of a credit network, therefore, as the chance that a random transaction will succeed in a random configuration.  However, this measurement will depend heavily on choices of transaction and configuration distributions.  For a real-world system settling real-world transactions, the relevant distribution on configurations is the distribution that arises from performing transactions drawn from the real-world exogenous transaction distribution.  

Consider, then, the Markov chain on the space of configurations where at each step, a payer vertex $x$ and a payee vertex $y$ are chosen with probability proportional to $\lambda_{xy}\in \mathbb{R}^{\geq 0}$, and $x$ pays one unit of money to $y$ if possible.  The stationary distribution thus gives the relevant configuration distribution.

\begin{definition}[Liquidity]
The {\em liquidity} of a credit network between vertex $x$ and $y$ is the probability that there exists a directed path from $x$ to $y$ of capacity at least 1 in a configuration drawn from the stationary distribution of the induced Markov chain.
\end{definition}

The distribution can be complicated and depends on exact transaction rates.  However, if rates are {\em symmetric} ($\lambda_{xy}=\lambda_{yx}$), then the stationary distribution of the Markov chain is uniform over all reachable cycle-equivalence classes (Theorem 5, \cite{dandekar2011liquidity}).  For the rest of this paper, we will assume that a unique stationary distribution exists; this happens if there are not two sets of agents that never transact with each other.  Liquidity analysis thus reduces to counting these classes.

The following definition will be useful for the rest of the discussion:

\begin{definition}[Score Vector]
The {\em Score Vector} of a configuration $C$ is $S(C)=\lbrace S(C)_v=\sum_u w(v,u)\rbrace_{v\in V}$.  Equivalently, $S(C)_v$ is the weighted outdegree of $v$ in $C$.

\end{definition}

For convenience, we will write $S_v$ to mean $S(\cdot)_v$ when either the specific configuration is clear from context or when we refer to a large class of configurations.

A successful transaction always decreases the score of a payer and increases score of the payee by the same amount.   When the payment is along a cycle, then, the score vector is invariant.  Hence, for a given credit network, a score vector uniquely captures a cycle-equivalence class, and Kleitman and Winston \cite{kleitman1981forests} show that the number of score vectors on a graph is equal to the number of forests of that graph (where a forest is an acyclic subset of edges).  Furthermore, Proposition 2.1 of \cite{goel2015connectivity} shows that the number of cycle-equivalent states where $x$ can pay $y$ is equal to the number of forests that place $x$ and $y$ in the same connected component.   The liquidity analysis of \cite{goel2015connectivity} crucially relies on this correspondence.

\section{Constrained Credit Networks}

\subsection{Node Constraints}

The credit network models transactions in a real-world trust network.  However, the model only accounts for independent bilateral relationships.  But a lender might also care about a borrower's total outstanding obligations. Conversely, one agent might want to limit her total lending.  More generally, suppose that each individual in a graph $G=(V, E)$ wishes to limit her total lending to the other agents, in addition to her bilateral lending limits.  

Let the aggregate limit on an individual $v$ be $c_v$, and suppose that the network is constrained such that in every valid configuration, $\forall v~S_v\leq c_v$.  If $k_v$ is the score of $v$ in the initial configuration, when no debt notes have been issued, this constraint means $v$ is disallowed from issuing more than $c_v-k_v$ notes.

\begin{theorem}

Suppose that the credit network system is required to remain in cycle-equivalent classes where $\forall v$, $S_v\leq c_v$.  Then the following properties over the set of cycle-equivalent classes satisfying these constraints are maintained:

\begin{enumerate}
\item Route-Independence
\item Cycle-equivalence $\iff$ Transaction-equivalence (for reachable cycle-equivalent states)
\item Symmetric transaction distribution $\implies$ the stationary distribution of the induced Markov chain on reachable configurations is uniform.
\end{enumerate}

\label{thm:simplepreds}

\end{theorem}

\begin{proof}

Theorem \ref{thm:simplepreds} follows directly from Theorem \ref{thm:predicates}, to be proved later.  The following section gives an intuitive demonstration.

\end{proof}

This theorem shows that independent restrictions on node behavior preserve most useful properties of credit networks.  The main property lost is the correspondence between forests and cycle-equivalence classes.  However, constraints can provide additional structure that often more than makes up for this loss.

But first, we give a constructive proof of Theorem \ref{thm:simplepreds} by showing that the individual node constraints can be modeled by a standard credit network.

Intuitively, an agent $v$ borrowing money is akin to routing flow to $v$ in the graph.  The maximum amount of flow that can be routed into $v$, then, is the min-cut of the graph that isolates $v$.  To demonstrate Theorem \ref{thm:simplepreds}, we give a construction of a network gadget, illustrated in figure \ref{fig:aggregates}, that separates each agent from the others with a small min-cut while preserving the rest of the graph.

Let $G^\prime$ be built from $G=(V,E)$ by, for each vertex $v\in V$, adding a ``fake'' vertex $F(v)$ connected only to $v$ by a new edge of capacity $c_v$.  The starting configuration of the network will be the starting configuration of $G$, and $w(v, F(v))=k_v$.  When agents $x$ and $y$ transact, they route transactions from $F(x)$ to $F(y)$. Then the score of $F(v)$ in $G^\prime$ is the score that $v$ would have in $G$ if agents had used $G$ instead of $G^\prime$.  Because every transaction involving a vertex $v$ runs through $(v, F(v))$, $v$ cannot lend more than $k_v$ in total.

\begin{figure}

\centering
\caption{A credit network with aggregate node constraints, implemented using gadgets}
\includegraphics[width=8.4cm]{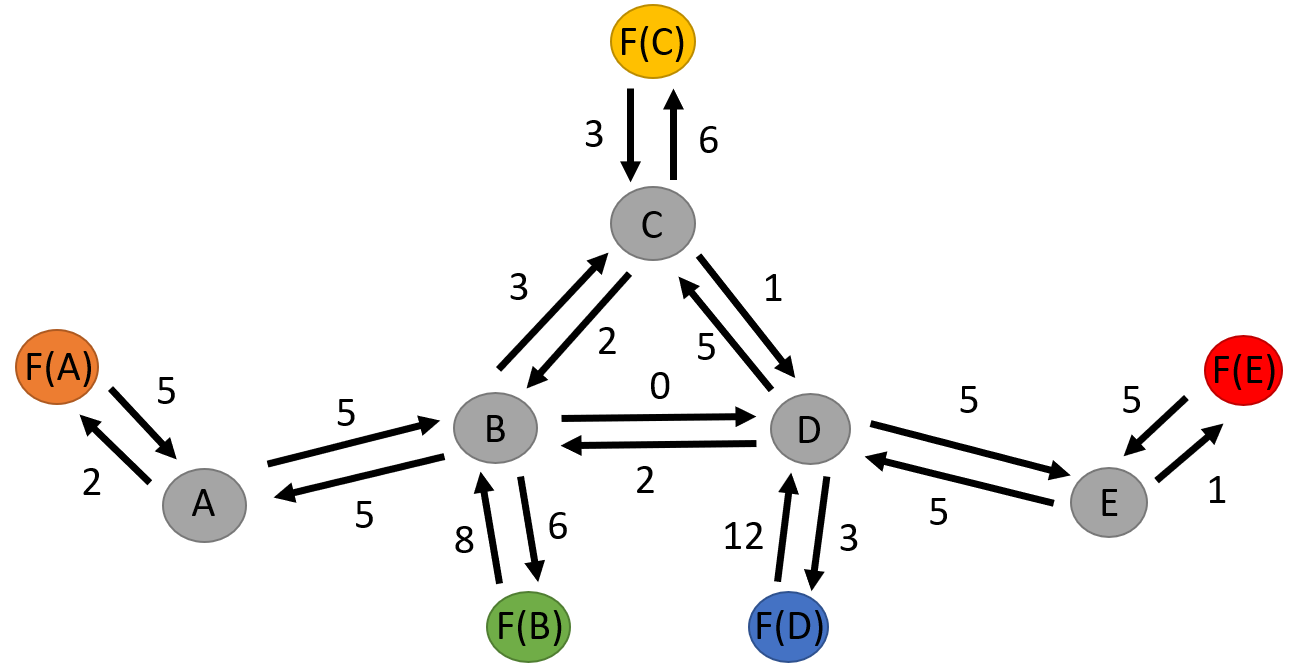}
\Description{A credit network where every vertex has an associated ``fake'' vertex connected only to it, which implements node constraints.}

\label{fig:aggregates}
\end{figure}

Observe that because no transactions originate from within $V$, in any collection of states reachable from each other, the score of every $v\in V$ is constant.  As such, conditioned on the choice of the starting configuration of the credit network, we can identify uniquely every reachable cycle-equivalent state with a score vector over only $F(V)$.  

Because this network operates exactly as a vanilla credit network in which half the vertices never perform transactions, the route-independence property and the property that transaction-equivalence is the same as cycle-equivalence on reachable states both hold. Additionally, a small modification of the proof of Theorem 5 of \cite{dandekar2011liquidity} shows that the Markov chain on this credit network starting at that start configuration is uniform over reachable score vectors.  

We now show that this combinatorial preservation is maintained in a more general, expressive notion of credit network constraint.

\subsection{Group Limits and Arbitrary Predicates}

Suppose that a business owner applies for a loan.  When assessing default risk, the lender would likely care about the individual's other lending or borrowing, as discussed previously.  However, a default of one member of a organization is likely correlated with the default of other members of that group.  As such, a lender might be concerned about the total borrowing of a group of agents, and might wish to impose an aggregate borrowing limit on the whole group.  Consider for example the credit network in \ref{fig:groups}, where agents in the blue box are not allowed to borrow too much in aggregate from the group outside the box.

Although we know of no network gadget for enforcing this property, satisfaction of the property can still be checked efficiently.

\begin{figure}

\caption{A credit network with an aggregate constraint on a group.  Here, the group in blue is not allowed to have its aggregate indegree, relative to vertices outside the group, exceed 12 (in this configuration, the aggregate indegree is 10).}
\centering
\includegraphics[width=8.5cm]{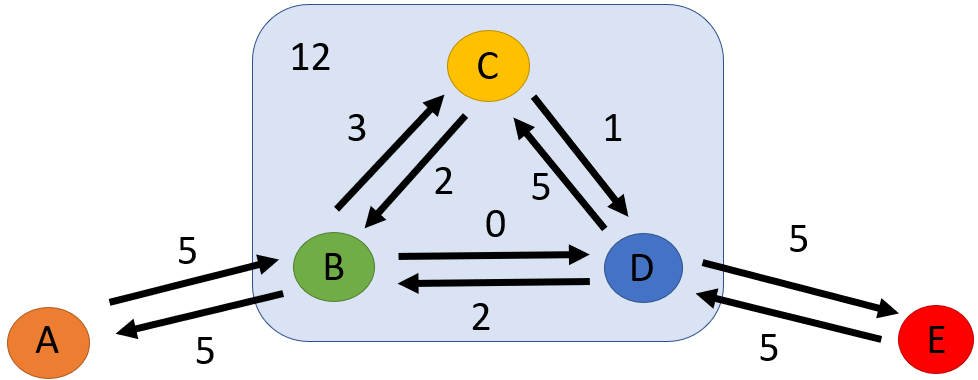}
\Description{A credit network where some nodes are grouped together by a group constraint.}

\label{fig:groups}
\end{figure}

More generally, we can study the dynamics of a credit network with broader lending restrictions imposed.  A network designer might like to require, for example, that agent $v_1$ can pay agent $v_2$ but only if it owes less than a certain amount to $v_3$.  

In fact, even when no gadget exists, any predicate that is well-defined on cycle-equivalence classes will preserve the properties in Theorem \ref{thm:simplepreds}.

\begin{definition}[Well-formed Predicate]
A Boolean predicate $P$ on configurations of a credit network is {\em well-formed} if, given cycle equivalent configurations $c_1, c_2$, $P(c_1)=P(c_2)$.  
\end{definition}

Note that the total amount that a group of nodes has borrowed from other nodes is invariant within a cycle-equivalent class.  Hence, restrictions on group aggregate borrowing, as in the above example, are well-formed predicates.  Of course, Boolean combinations of well-formed predicates are also well-formed predicates.\footnote{Not all ``natural'' predicates are strictly well-formed.   For example, ``A is willing to lend \$10 to B, but only if B's debt to C is less than \$5,'' as it relies on states of links that might vary within a cycle-equivalent class.  Such a constraint can be made well-formed by asking instead whether there exists a cycle-equivalent configuration satisfying the original constraint.  Informally, we conjecture that most constraints of interest fit within this framework and are computable with max-flow computations.} In fact, any Boolean function well-defined on the set of score vectors is a well-formed predicate.


\begin{theorem}
Given any well-formed predicate $P$ on the cycle-equivalence classes of a credit network, the following properties hold in the corresponding constrained credit network where $P$ is enforced:

\begin{enumerate}
\item Route Independence
\item Cycle-equivalence $\iff$ transaction-equivalence (for reachable cycle-equivalence classes)
\item Symmetric transaction distribution $\implies$ the stationary distribution of the induced Markov chain on reachable cycle-equivalence classes is uniform.
\end{enumerate}

\label{thm:predicates}
\end{theorem} 

\begin{proof} See Appendix A. \end{proof}
\section{Liquidity Analysis}

\label{sec:liquidityanalysis}
It now remains to study the impact that constraints have on liquidity of credit networks.  In particular, we will focus on credit networks with node constraints.  We also show an interesting combinatorial difference between constrained and unconstrained networks.  For all theorems in this section, transaction rates will be symmetric and that a unique stationary distribution over cycle-equivalence classes exist (that is, there are not disjoint sets of vertices that never transact with each other), so the distribution over cycle-equivalence classes is uniform (Theorem 5, \cite{dandekar2011liquidity})

Before beginning, note that if an arbitrary predicate can be evaluated efficiently, then liquidity can be estimated experimentally by simply simulating a sequence of random transactions, as in the Markov chain used to define liquidity.  The authors of \cite{dandekar2011liquidity} used such simulations to conjecture liquidity in analytically intractable classes of (unconstrained) credit networks; the caveat here is that these Markov chains lack general but meaningful mixing time bounds.

\subsection{Trees}

Liquidity can be exactly computed in credit networks that have a simple structure.  In particular, when subjected to node constraints, a natural dynamic programming algorithm computes liquidity in a node-constrained tree.

\begin{theorem}
\label{thm:tree}
In a tree with node constraints, liquidity can be computed in time polynomial in the size of the graph and in the maximum capacity along an edge.
\end{theorem}

In fact, this algorithm extends to graphs that are close to being tree-like, in the sense of having low treewidth \cite{bodlaender1998partial}.  \footnote{For a detailed explanation of a similar algorithm, see \cite{noble1998evaluating}}

\begin{theorem}
Suppose a graph $G=(V, E)$ with node constraints has tree-width $k$, and let $S=\max_v \sum_{u\in \Gamma(v)}c(v,u)$.  There exists an algorithm to compute liquidity in time $poly(\vert V\vert, S^k, 2^{k^2})$.  
\label{thm:treewidth}
\end{theorem}

\begin{proof} See Appendices \ref{prf:tree} and \ref{prf:treewidth}\end{proof}

\subsection{Star Graph}

Moving from liquidity computation to liquidity analysis, consider as a first example the class of star graphs, where every edge runs from an external vertex $v_i$ to a the central vertex $u$ with capacity $c_i$.  This type of graph will be useful later, when we show that nontrivial node constraints make constrained credit networks on arbitrary graphs functionally equivalent to a constrained star.

Without loss of generality, observe that the edge nodes need no extra constraint, as any extra constraint is equivalent to a decrease in capacity on the associated edge.  Similarly, the center vertex can be constrained to perform no transactions whatsoever, if we add an extra outside vertex and have any transactions involving the center go to this vertex, as in figure \ref{fig:aggregates}.

\begin{theorem}

When a star graph is constrained such that its central vertex $u$ has score $\lfloor s_u=\Sigma_{i}c_i/2\rfloor$, the steady-state failure probability between any two vertices $i$ and $j$ is at most $4/(c_i+c_j)$.  Moreover, the steady-state failure probability is at least $2/(c_i+c_j+2)$.  

\label{thm:star}

\end{theorem}

\begin{proof}See Appendix \ref{prf:star}\end{proof}

In an unconstrained star, the failure probability is $1/(c_i+1) + 1/(c_j+1) - 1/((c_i+1)(c_j+1))$, so constraining the star has not significantly reduced liquidity.

\subsection{Expander Graphs}
Liquidity is intuitively related to the edge-expansion of the underlying graph.  After all, route-independence means that transactions in a credit network are akin to single-commodity flows in a directed graph.  Such a flow will fail if it hits a bottleneck, such as if the source or destination is in a poorly-connected set of vertices, or if the graph's min-cut is small.  

More specifically, a transaction of size $k$ from $u$ to $v$ will fail if and only if there is some cut of the graph separating $u\in A\subset V$ from $v\in B\subset V$ such that $k$ units of flow cannot move from $A$ to $B$.  From another viewpoint, the collective capacity of the edges from $A$ to $B$ gives an implicit group constraint on $A$ (and $B$).  A predicate constraint on $A$ will thus not change network behavior unless it supersedes this implicit constraint.  The interesting parameter regimes for analysis, therefore, are those where many sets of nodes have superseding group constraints. 

Rather than study arbitrarily-valued group constraints on every subset of vertices, we will study node constraints; constraints on single vertices collectively give constraints on every set of vertices.  As such, one interesting parameter regime for node constraints is where the range of a node's net balance is decreased from its degree to below the edge expansion of the graph.  

Let $h(G)$ be the edge expansion of a graph $G=(V,E)$.  We will take here the edge expansion of a graph to mean \\$\min_{S\subset V : 0<\vert S\vert \leq \vert V\vert /2} \partial(S)/\vert S\vert$, where $\partial(S)$ is the total capacity of edges leaving $S$.  Let $d(v)$ be the weighted degree of a vertex $v$.

\begin{theorem}
Let $G=(V, E)$ be a credit network, and for each $v\in V$, let $0\leq r(v)\leq h(G)$ be some integer.  If every node's score is restricted to lie between $(d(v) - r(v))/2$ and  $(d(v) + r(v))/2$, then the constrained credit network is equivalent to the star credit network $H$ where the central node $u$ has a fixed score of $S_u=\Sigma_{v\in V}r(v)/2$ and $c(v_i, u)=r(v)$.\footnote{This theorem requires two minor technical restrictions.  First, that $r(v)$ is even if and only if $d(v)$ is even; if $r(v)$ is odd but $d(v)$ is even, reducing $r(v)$ by one does not change any behavior of the credit network.  Second, that $\Sigma_{v\in V}r(v)$ is divisible by $2$.}  
\label{thm:expander}
\end{theorem}

\begin{corollary}
When $r(v)=\lfloor h(G)\rfloor $ for all $v$, a constrained expander graph has, for every pair of vertices, liquidity between $1-1/(\lfloor h(G)\rfloor + 1)$ and $1-2/\lfloor h(G)\rfloor$.

\end{corollary}
\begin{proof} See Appendix \ref{prf:expander}\end{proof}

For comparison, \cite{goel2015connectivity} shows that the average liquidity in an unconstrained graph is at least $1-2/h(G)$, but the proof requires several more pages of analysis.  What is surprising here is that the liquidity between two specific vertices does not particularly depend on details of edge connections to those vertices within the graph.  $h(G)$ is rounded because edge capacities are integral.

\subsection{Monotonicity}
A credit network can be thought of as a specialized network for transferrring financial commodities.  Some kinds of networks for commodity transfer, such as road traffic, are subject to what is typically referred to as "Braess's Paradox" \cite{braess2005paradox}.  In short, the paradox is the observation that adding roads in a road network can reduce the overall throughput of the network, when drivers choose routes selfishly.

It would be highly undesirable if this paradox existed in Lightning-style credit network implementations.  In the anonymous world of cryptocurrencies, there is no obvious feasible manner of implementing anything other than selfish route choices.  It could be possible, then, that bad actors could attack the network and, for example, drive up transaction prices.

We would like, therefore, to prove that the addition of an edge to a graph will never decrease liquidity.  When studying unconstrained credit networks, the authors of \cite{goel2015connectivity} call this the ``monotonicity conjecture'' and show that this notion is equivalent to the well-studied negative correlation conjecture on graphical matroids; for more on that conjecture, see \cite{semple2008negative} and \cite{choe2006rayleigh}.

One could directly generalize the monotonicity conjecture to constrained credit networks.  Specifically, one might hope that the addition of an edge in a network containing constrained stars will not decrease the liquidity between any two points.  However, this notion is false.

\begin{example}
\label{example:nonmono}

Let $G$ be a star graph with four endpoints $v_1, v_2,v_3,$ $ v_4$ and center $u$, where the capacity of each edge is $1$ and the score of $u$ is constrained to be $2$.  Then the liquidity between any two endpoints is $1/3$.  Let $H$ be formed an edge between $v_3$ and $v_4$ of capacity $n$.

In $H$, the liquidity between $v_1$ and $v_2$ is $(n+2)/(4n+6)$, which is decreasing in $n$ and always less than $1/3$.  

\end{example}

Graph-like objects that disobey a monotonicity-like conjecture, especially simple examples, are rare.  It would be interesting to understand how constraints enable this qualitative change in behavior.

Although a direct analogue of a monotonicity conjecture is violated, note that liquidity still respects the bounds implied by Theorem \ref{thm:star}.  In particular, Theorem \ref{thm:star} implies that the liquidity between $v_1$ and $v_2$ is at least $1/4$, and for all $n\geq 0$, $(n+2)/(4n+6)>1/4$.  We conjecture, therefore, that while liquidity might decrease as in Example \ref{example:nonmono}, the bounds implied by Theorem \ref{thm:star} are never broken.

Clearly, if there exists some pair of vertices $x$ and $y$ in subgraph of expansion $\beta$ with pairwise liquidity $\alpha>1-1/\beta$, then this replacement must decrease their pairwise liquidity.  More generally, in such a graph, if a vertex $x^\prime$ is connected to only $x$ by an edge of capacity $1$, and a vertex $y^\prime$ is connected only to $y$ by an edge of capacity $1$, then the liquidity before replacement is $\alpha/4$, while after the replacement it is at least $(1-2\beta)/4$. 

More specifically, we conjecture that a multiplicative liquidity reduction by a factor of $1-2\beta$ is the worst reduction that can happen.


\begin{conjecture}
Let $G=(V,E)$ be a credit network, let $S\subset V$ with subgraph expansion $h_S(G)$, and let $H$ be the credit network formed by replacing $S$ in $V$ with a constrained star, as in Theorem \ref{thm:star}.  Then for all $u, v\in V$, the liquidity between $u$ and $v$ in $H$ is at least the liquidity of between $u$ and $v$ in $G$ multiplied by $1-2/h_S(G)$.  

\label{conjecture:monotonicity}
\end{conjecture}


\section{Cryptocurrency Applications and Future Work}

Cryptocurrency innovations like the Lightning network on Bitcoin rely on a network of bilateral payment channels that behaves like a credit network.  However, in the trustless cryptocurrency context, manufacturing the trust required for a payment channel requires committing capital into escrow per edge, which has a cost.  Higher maintenance costs generally mean higher transaction costs for users.  The network, therefore, is incentivized to find a design that gives a good tradeoff between liquidity and total escrow costs.

As shown in Section \ref{sec:liquidityanalysis}, a star-like design where every agent has a global lending limit achieves the optimal tradeoff between liquidity and total escrow costs.  In practice, this would look roughly like a large, permissioned smart contract.

For simplicity of example, suppose that every agent has $\$D$ to invest into $D$ edges, each with capacity 2 (so each agent is initially responsible for $D$ units of escrow).   Let the expansion of the resulting graph $G$ be $h(G)$.

Two agents transacting {\it only} with each other could {\it at best} get liquidity $1-1/2D$ (and liquidity $0$ with all others).  Using a standard Lightning system, agents only get on average pairwise liquidity $1-2/h(G)$.  By switching to a multiparty contract, {\it every} pair of agents can achieve liquidity $1-2/D$ - the asymptotically optimal tradeoff.  
The exact savings will vary by graph, but $h(G)$ can be much smaller than $D$.  
Furthermore, routing across a multiparty contract is trivial, and routes are valid unless the sender is bankrupt.


\subsection{Future Work}

Section 3 shows that adding constraints preserves most useful properties of credit networks.  However, it does eliminate the correspondence between the forests and score vectors.  Unfortunately, the bijection in \cite{kleitman1981forests} is algorithmic and, other than in tree-like graphs, an analysis of the constrained credit network using forests is not obvious.  We leave this as an area of future research.

In section 5, the escrow savings require assumptions on the transaction distribution and might disappear if Conjecture \ref{conjecture:monotonicity} were false.  We leave as future work an experimental analysis of real-world Lightning networks, particularly with regard to the tradeoff between subgraph expansion, escrow savings, and implementation concerns.  We would also like to understand how more realistic distribution assumptions would affect our results.  Similarly, many real-world implementations of payment networks have strategic agents that resettle network links via on-chain transactions when, for example, a net balance grows too much (e.g. \cite{branzei2017charge}).  We would like to understand how these resettlement policies 

In \cite{dandekar2015strategic}, Dandekar et al consider the strategic formation of credit networks under a model of balancing  liquidity against exposure to defaulting trade partners.  Constraints allow for many interesting scenarios in which to study the behavior of rational agents.  For example, a node constraint is equivalently a guarantee that one will not borrow more than a total amount.  In some contexts, guarantees like this on nodes or groups could lead to larger bilateral lines of credit.  Understanding the incentives at play could improve designs of credit network-like systems.


\section{Conclusion}
The credit network is a model for transactions across a network of agents.  Initially studied in contexts related to social networks, the model forms a close abstraction of modern cryptocurrency ``Layer 2'' protocols like Lightning that are being deployed across the internet.
However, the credit network is limited in its ability to describe agent behavior.  We study the effects of constraining the behavior of agents in a credit network beyond the implicit constraints in a credit network, or alternatively, the effects of limited guarantees on agent solvency.  In particular, these constraints preserve the combinatorial structure of credit networks.  Aggregate node-based borrowing constraints transform complicated graphs into simple stars, showing that the details of graph structure ultimately are of little significance.  These constraints also enable modeling of more interesting node behavior, and moreover, assuming Conjecture \ref{conjecture:monotonicity}, the reduction from complex graphs to star graphs achieves the optimal tradeoff between liquidity and escrow costs in a Lightning-style payment network.

\bibliographystyle{ACM-Reference-Format}
\bibliography{unified-bib}

\section*{Appendix}

\appendix
\section{Proof of Theorem \ref{thm:predicates}}

\begin{proof}
First, note that a well-formed predicate cannot distinguish representatives of a cycle-equivalent class, so the equivalence relation retains a well-defined structure.  We also note that Boolean combinations of well-defined predicates are well-defined predicates.

\begin{enumerate}
\item

  Without the predicate, given two feasible routes $p_1$ and $p_2$ from $x$ to $y$ and some configuration $c$, $p_1(c)$ is cycle-equivalent to $p_2(c)$.  As a predicate can only declare that certain cycle-equivalence classes are valid or invalid, and thus not comment on the intermediate routing computation steps, its evaluation must only depend on the resulting cycle-equivalent state.  Hence, the choice of route does not affect the result.

\item

Suppose $c_1$ and $c_2$ are cycle-equivalent configurations.  Then given that a well-formed predicate cannot distinguish cycle-equivalent states, a predicate can only invalidate a transaction on $c_1$ if and only if it invalidates the transaction on $c_2$.  Hence, cycle-equivalence implies transaction-equivalence.

Conversely, let $C$ be a set of cycle-equivalence classes that are all reachable from each other given a well-formed predicate, and let $c_1$ and $c_2$ be distinct cycle-equivalence classes.  Then there is a sequence of transactions $\tau$ that takes $c_1$ to $c_2$.  Observe that $\tau$ must decrease the score of at least one vertex, and the score of any vertex is bounded below by $0$.  Then there exists a $k$ such that $\tau^k(c_1)$ is a valid transaction but $\tau^{k+1}(c_1)$ is not.  But $\tau^{k+1}(c_1)=\tau^k(c_2)$.  Hence $c_1$ and $c_2$ are not transaction-equivalent.   

Without the reachability requirement, a predicate could, for example, invalidate all transactions.  Then non cycle-equivalent states would be vacuously transaction-equivalent.

\item

The same argument used in Theorem 5 of \cite{dandekar2011liquidity} holds here.  For completeness, we include the proof.  

Let $\tau_{ij}$ be the transactions taking $C_i$ to $C_j$, and let $P(C_i, C_j)=\sum_{(s,t) \in \tau_{ij}}\lambda_{st}$.  This generates a symmetric matrix where the sum of all the entries in each row and column is a constant, so this matrix normalized is a symmetric stochastic transition matrix.  In fact, this matrix is the transition matrix of the Markov chain used to define liquidity.  Since this transition matrix is a stochastic symmetric matrix, the uniform distribution is stationary.

Note that a Markov chain on a set of states reachable from each other is by definition irreducible.  Since transactions are invertible, there are no sink states.  And we can always make the chain aperiodic by adding a $1/2$ probability of staying in place.  Hence, the Markov chain is ergodic.

\end{enumerate}
\end{proof}

We note that it is important for the proof of $(2)$ that we look only at cycle-equivalence classes reachable from each other.  Suppose that one predicate evaluates to $1$ only on a small number of isolated configurations.  Then in every valid state, no transaction is viable, so all states are transaction-equivalent.

\section{Algorithms for Computing Liquidity in Trees}
\subsection{Proof of Theorem \ref{thm:tree}}
\label{prf:tree}
Consider the following dynamic programming algorithm for computing the liquidity between two vertices in a tree where nodes have aggregate constraints.


Pick some vertex $r$ to be the root of the tree.  Let $p(v)$ denote the parent of $v$, and $q_i(v)$ the $i$th child of $v$.  Let $d(v)$ be the number of children of $v$.  

Let $C(v, k)$ be the number of configurations in the subtree (satisfying all subtree constraints) rooted at $v$ such that $w(p(v), v))=k$.  It suffices to show how to compute $C(v, \cdot)$ given access to $C(q_i(v), \cdot)$ (if $q_i(v)$ is a leaf vertex, we say $C(q_i(v), \cdot)=1$).

Let $D_i(v, k)$ be the number of configurations of the subtree consisting of $v$ and the first $i$ child subtrees such that the score of $v$ is $k$.  Then $D_{i+1}(v, k)=\sum_{j=0}^{c(v, q_{i+1}(v))}D_i(v, k-j)*C(q_{i+1}(v), j)$, and let $D_0(\cdot, \cdot)=1)$.

Let $X_v$ be the set of scores of $v$ that are considered valid by the constraints.

Then $C(v, k)= \sum_{i\in X_v}D_{d(v)}(v, i - (c(p(v), v)-k))$, where $D_{d(v)}$ is $0$ if not defined otherwise above.

The total number of configurations is thus $\sum_{i\in X_r}D_{d(r)}(r, i)$.

The above algorithm iterates over, at each vertex, all possible ways such that that vertex has a specific score.  This can be extended to also satisfy some simple local predicates by altering this iteration to exclude particular classes of configurations.

For example, to track liquidity, note that there is a unique path from $u$ to $v$.  As the algorithm walks across the graph, it can simply throw out configurations where an edge on this path is entirely oriented in the wrong direction.

Let $H$ be the maximum capacity of an edge.  

Storing $C$ requires $O(\vert E\vert H)$ entries and, given oracle access to $C$ on the children of a vertex $v$, requires computation time $O(d(v)H)$.  So running the entire algorithm requires time polynomial in $H$ and the size of the graph.

\subsection{Proof of Theorem \ref{thm:treewidth}}
\label{prf:treewidth}

Let $G=(V, E)$ be some graph of treewidth $k$, and let $(X_i, (I, F))$ be a nice tree decomposition of $G$, as described in Theorem 7 of \cite{bodlaender1998partial}.  Without loss of generality, assume that the leaf vertices have no constraints, and that non leaf vertices $v$ are constrained to have fixed score $s_v$.  

Iterating from the leaves to the root of the tree decomposition, the algorithm maintains a complete list of score vectors of active vertices and a count of the number of ways that the induced subtree can produce this score vector on the active vertices (not including edges between active vertices) while satisfying the node constraints of inactive vertices in the subtree.

At a leaf node, the only score vector is $(0)$, which occurs with multiplicity $1$.  

At an introduce node $X_i$ with child $X_j$ that introduces a node $v$, the algorithm simply adds an entry of $0$ to every active score vector corresponding to the score of $v$.  This maintains the induction invariant, as there are no edges between the introduced vertex and the inactive vertices of the subtree.

At each join node with children $X_i$ and $X_j$ (with $X_i=X_j$), the algorithm iterates over every pair of score vectors, with one from $X_i$ and one from $X_j$, adding score vectors coordinatewise and multiplying multiplicities.  The resulting list is then de-duplicated, with multiplicities added as necessary.  Join nodes join disjoint subtrees, so score vectors of disjoint subconfigurations add coordinatewise, and there is no double-counting of edges (since edges between active vertices are not yet accounted for).

At a forget node $X_i$ that forgets a node $v$ (with child $X_j=X_i\cup\lbrace v\rbrace$), let $Y$ be the edges between $v$ and the vertices of $X_i$.  For every active score vector $s$ of $X_i$, the algorithm computes a list of potential score vectors of $X_i$ based on $s$, by iterating over all ways of directing the edges of $Y$.  These score vectors are de-duplicated, and their multiplicities are set to be the multiplicity of $s$.  The algorithm then collates all these lists and de-duplicates them, adding multiplicities when necessary.  The algorithm then asks if $v$ satisfies the constraints.  If yes, it preserves that vector, with the entry for $v$ removed.  Otherwise, it drops that vector.  The algorithm then merges duplicate score vectors, adding multiplicities when duplicates occur.

Furthermore, since satisfaction of each vertex's constraint is based only on the score of an individual vertex, and once a vertex is ``forgotten,'' all of its edges have been accounted for in a particular configuration $c$ of the subtree, if a constraint is satisfied when its vertex is forgotten, it will be satisfied in any configuration that extends $c$ to the entire graph.

 At the final node of the tree decomposition, then, the algorithm is left with an empty score vector and a count of the number of score vectors satisfying all the vertex constraints.
 
Let $I$ be the number of states tracked at any node.  Then evaluating an introduce node takes time $O(I)$, a join node takes time $O(I^2)$, and a forget node takes time $O(S^k I + I^2)$.  There are at most $O(k\vert V\vert)$ nodes, so the entire algorithm takes time $O(k\vert V\vert (S^k + I^2))$.  
 
Let $S$ be the maximum (unconstrained) score of any vertex.  Then the number of score vectors tracked at any one node is at most $S^k$.  Hence, the entire algorithm takes time $O(k\vert V\vert (S^k + S^{2k})$

This algorithm computes the number of score vectors satisfying the node constraints.  To compute liquidity from a vertex $u$ to a vertex $v$, the algorithm needs to also track connectivity patterns between vertex.  

In particular, with each score vector, the algorithm also associates a directed graph on the set of active vertices, where an edge from $x$ to $y$ in this graph signifies that in subtree configuration that produces this score vector and directed graph, there is a directed path from $x$ to $y$.

Leaf nodes start with an empty connectivity graph with only a single vertex.  Introduce nodes add a disconnected node to the connectivity graph.  Forget nodes, when iterating over arrangements of edges between the forgotten vertex and active nodes, compute which vertices gain new connections and update the graph accordingly, while dropping the forgotten vertex from the connectivity graph unless the forgotten vertex is one of $u$ or $v$.  Join nodes, when pairing subconfigurations from different subtrees, simply take the union of paired connectivity graphs.  This maintains the connectivity patterns throughout the computation.  At the end, the algorithm is left with connectivity graphs containing only $u$ and $v$, and the number of score vectors satisfying the constraints that generate each connectivity pattern.  Liquidity follows with a little arithmetic.

This increases the size of the state space at each node of the tree-decomposition by at most a factor of $2^{(k+1)k}$.  So the entire algorithm takes time $O(k \vert V\vert S^{2k}2^{2(k+1)k}))$.  These time complexity bounds are likely not tight.

\section{Proof of Theorem \ref{thm:star}}
\label{prf:star}



\begin{proof}
In the unconstrained star graph, the number of configurations where $w(v_i, u)=k$ for some $v_i$ would be constant for all feasible $k$.  In this situation, this may not be the case, but in fact, if the score of the center vertex is half its (unconstrained) maximum, a symmetry argument shows that the number of cases where $w(v_i, u)=k$ and $w(v_i, u)=c(v_i, u)-k$ are equal.  

Let $n$ be the number of external vertices, and let $G_i$ be the star graph consisting of only the central vertex and the first $i$ vertices.  Let the number of configurations of $G_i$ in which $S_u=k$ be $C(G_i, k)$, and let $M_i=(\sum_{j=1}^i c_i)/2$.  Then for any $k$,  $C(G_{i+1}, k)=\sum_{j=0}^{c_{i+1}} C(G_i, k-j)$.  Observe that if $C(G_i, \cdot)$ is maximized at $M_i$ and symmetric about $M_i$, then $C(G_{i+1}, \cdot)$ is maximized at $M_i + (c_{i+1}/2) = M_{i+1}$ and symmetric about $M_{i+1}$ (if $M_i$ was rounded down, then $C(G_i, M_i)=C(G_i, M_i+1)$ and $C$ is symmetric about $M_i+1/2$).  

Since this symmetry and maximization in the middle holds for $i=2$, induction shows this holds for $i=n$.  Let $M=\sum_i c_i/2$.

Now, for any two vertices $v_i$ and $v_j$, for any configuration, let $\Delta_k=w(v_i, u) + w(v_j, u)$.  
Let $H$ be the unconstrained graph without vertices $v_i$ and $v_j$, let $A_k$ be the number of states of $H$ such that $S_u=M-\Delta_k$, and let $B_k=\vert A_k\vert$.  

Consider the one-to-many map from $A_k$ to states of $G$ that completes a state in $A_k$ to a state of $G$ by choosing any values for $w(v_i, u)$ and $w(v_j, u)$ such that $w(v_i, u)+w(v_j, u)=\Delta_k$.  For $\Delta_k\leq (c_i+c_j)/2$, there are $\Delta_k+1$ choices, and for $\Delta_k\geq (c_i+c_j)/2$, there are $(c_i+c_j)-\Delta_k+1$ choices.  Furthermore, note that in exactly one of these choices is $v_i$ bankrupt.  Note that these maps have well-defined inverses and the images of each $A_k$ are disjoint.

Now, consider the case where $0\leq \Delta_k\leq (c_i+c_j)/2$.  

Let $T=(c_i+c_j)/2$.  
The organization of states above enables counting the number of state where transactions fail and number of transactions overall.   As such, the probability $p$ of a failed transaction (conditioned on $0\leq \Delta_k\leq T$) is therefore simply the weighted summation $(\sum_{i=0}^T B_i)/(\sum_{i=0}^T (i+1)B_i)$.  By Chebychev's sum inequality, $1/p \geq ((\sum B_i)/(\sum B_i)) (1/T) (\sum_i (i+1))=(T+2)(T+1)/2T\geq T/2$. so the probability of transaction failure is at most $2/(T)$.  By a symmetricity argument, the same result holds when $(c_i+c_j)/2\leq \Delta_k\leq (c_i+c_j)$, so the result holds overall.

Conversely, the probability of transaction failure is a weighted summation of the probability of transaction failure conditioned on a particular $\Delta_k$, and thus must be at least the minimum of these conditional probabilities, which is $1/(T+1)$.

\end{proof}

If the score of the central vertex is constrained to be less than half its maximum, but close to half, then a (crude) bound on failure probability can be obtained by decreasing (artificially for analytical purposes) some of the capacities of the edges until the score is half of the reduced maximum.  The same holds for central scores larger than half the maximum.

The above recurrence relation in effect counts the number of ways in which to put $S_u$ indistinguishable items into $n$ boxes of varying sizes $c_i$.  When the $c_i$s are constant, the number of states is simply the generalized binomial coefficient.

\section{Proof of Theorem \ref{thm:expander}}
\label{prf:expander}
\begin{proof}

Let $G$ be a credit network with edge expansion $h(G)$, constrained to ensure that for all vertices $v$, $s_v\in ((d(v)-r(v))/2,$$ (d(v)+r(v))/2)$.


In any configuration of a credit network, it is impossible for vertex $v$ to pay vertex $u$ if and only if there exists a partition of the graph into a set $A$ and $B=V\setminus A$ such that $v\in A$ and $u\in B$ and all edges $(a,b)$ between $B$ and $A$ satisfy $w(a,b)=c(a,b)$, $w(b,a)=0$.  

Let $u,v\in V$, and let $A$ and $B$ be any partition of $V$ separating $v$ and $u$ such that the cut prevents $u$ from paying $v$.  Without loss of generality (by symmetry of the constraints), suppose $\vert A\vert \leq \vert B\vert$.  

Suppose there are $x\geq \vert A\vert h(G)$ edges pointing into $A$.  Then the number of edges contained within $A$ is 
$(\sum_{v\in A}d(v)-x)/2$.  Hence, the sum of the scores of every vertex in $A$ is $\sum_{v\in A}d(v) + x/2$.  Because $x\geq h(G)\vert A\vert$, the sum of scores in $A$ exceeds the aggregate bound implied by the individual bounds on all vertices in $A$, which is a contradiction.  

Hence, the only constraints that affect whether vertex $u$ can pay vertex $v$ are the per-vertex constraints we imposed on top of the credit network.  Particularly, each vertex is constrained to deviate only at most $r(v)/2$ from a score of $d(v)/2$.  

Let $H$ be a star network with a new vertex $u$ in the center, where $S_u$ is constrained to be constant (in fact, $S_u=\Sigma_{v\in V}r(v)/2$), and let $c(v_i, u)=r(v)$.  Note that for any score vector of $H$, adding $(d(v)-r(v))/2$ to each vertex's score gives a vector in $G$ satisfying the constraints on $G$.  Then any transaction in $G$ is feasible if and only if the corresponding transaction is feasible in $H$, and moreover, this correspondence between score vectors is maintained under transactions.

Thus, with regards to transaction feasibility, $G$ is equivalent to $H$.

We note that further constraining a vertex in $G$ only shrinks the capacity of the corresponding edge from that vertex to the center.

\end{proof}

\end{document}